\documentclass[letterpaper, 10 pt, conference]{_Aux/ieeeconf}
\IEEEoverridecommandlockouts
\usepackage{times}
\usepackage{setspace}
\usepackage{url}
\spacing{1}
\usepackage[utf8]{inputenc}
\usepackage[T1]{fontenc}
\usepackage{graphicx}		% For \includegraphics
\usepackage{wrapfig}
\usepackage[format=plain,font=footnotesize,labelfont=bf,labelsep=period]{caption}
\usepackage{sidecap} % For sidecaptions for figures
\usepackage{subfig}
\usepackage[export]{adjustbox}
\usepackage{float}

\usepackage{amsmath} % assumes amsmath package installed
\usepackage{amssymb}  % assumes amsmath package installed
\usepackage{amsthm}
\usepackage{mathtools}
\usepackage[normalem]{ulem}
\usepackage{paralist}	% For enumerations with better titled numbers
\usepackage[space]{grffile} % Space in filenames
\usepackage{color}
\let\labelindent\relax
\usepackage{enumitem}
\usepackage{bm}
\usepackage{cancel}
\usepackage{hhline}

\newtheorem{theorem}{Theorem}
\newtheorem{corollary}{Corollary}

\theoremstyle{definition}
\newtheorem{definition}{Definition}
\theoremstyle{remark}
\newtheorem{remark}{Remark}
\theoremstyle{definition}

\theoremstyle{definition}

\newtheorem{example}{Example}

\newcommand{\Sc}{\mathcal{S}}

\newcommand{\R}{\mathbb{R}}

\newcommand{\K}{\mathcal{K}}

\newcommand{\Hs}{\mathcal{H}}

\definecolor{blue}{RGB}{38,38,134}
\definecolor{darkblue}{RGB}{0,0,102}
\definecolor{lightblue}{RGB}{77,77,148}

\definecolor{gold}{RGB}{234, 170, 0}
\definecolor{metallic_gold}{RGB}{139, 111, 78}

\renewcommand{\cal}[1]{\mathcal{ #1 }}

\newcommand{\der}[2]{\frac{\mathrm{d} #1 }{\mathrm{d} #2 }}

\DeclareMathOperator*{\argmin}{argmin}

\usepackage{xfrac}

\usepackage{dblfloatfix}

\usepackage{cite}
\usepackage[dvipsnames,table,xcdraw]{xcolor}

\usepackage{amsfonts}
\usepackage{algorithmic}
\usepackage{textcomp}
\usepackage{multirow}
\usepackage{epsfig} % for postscript graphics files

\begin{document}

\title{ \bf
%Improved Safety Guarantees with Flexible Barrier Functions
%Generalizing Robust Control Barrier Functions for Improved \\ Safety Guarantees under Uncertainty 
Parameterized Barrier Functions \\ to Guarantee Safety under Uncertainty 
}

\author{Anil Alan$^{1}$, Tamas G. Molnar$^{2}$, Aaron D. Ames$^{2}$, G\'abor Orosz$^{1}$
\thanks{*This research is supported in part by the National Science Foundation (CPS Award \#1932091), Dow (\#227027AT) and Aerovironment.}%
\thanks{$^{1}$A. Alan and G. Orosz are with the Department of Mechanical Engineering, and G. Orosz is also with the Department of Civil and Environmental Engineering, University of Michigan, Ann Arbor, MI 48109, USA,
{\tt\small anilalan@umich.edu, orosz@umich.edu}.}%
\thanks{$^{2}$T. G. Molnar and A. D. Ames are with the Department of Mechanical and Civil Engineering, California Institute of Technology, Pasadena, CA 91125, USA,
{\tt\small tmolnar@caltech.edu, ames@caltech.edu}.}%
}

\maketitle
\thispagestyle{empty}         % Removes the page number in the first page
\pagestyle{plain}

\begin{abstract}
Deploying safety-critical controllers in practice necessitates the ability to modulate uncertainties in control systems.  
In this context, robust control barrier functions---in a variety of forms---have been used to obtain safety guarantees for uncertain systems.
Yet the differing types of uncertainty experienced in practice have resulted in a fractured landscape of robustification---with a variety of instantiations depending on the structure of the uncertainty.
This paper proposes a framework for generalizing these variations into a single form: \emph{parameterized barrier functions (PBFs)}, which yield safety guarantees for a wide spectrum of uncertainty types.
This leads to controllers that enforce robust safety guarantees while their conservativeness scales by the parameterization. 
To illustrate the generality of this approach, we show that input-to-state safety (ISSf) is a special case of the PBF framework, whereby improved safety guarantees can be given relative to ISSf.  
\end{abstract}

%%%%%%%%%%%%%%%%%%%%%%%%%%%%%%%%%%%%%%%%%%%%%%%%%%%%%%%%%%%%%%%%%%%%%%%%%%%%%%%%%%%%%%%%%%%%%%%%%%%%%%%%%%%%%%%%%%%%%%%%%%%%%%%%%%%%%%%%%%%%%%%%%%%%%%%%%%%%%%%%%%%%%%%%%%%%%%%%%%%%%%%%%%%%%%%%%%%%%%%%%%%%%%%%%%%%%%%%%%%%%%%%%%%%%%%%%%%%%%%%
\section{Introduction} \label{sec:intro}

Control barrier functions (CBFs) \cite{ames2017control} have become a popular tool for synthesizing safety-critical controllers due to their generality and relative ease of synthesis and implementation. Safety is encoded by a single scalar inequality constraint: ${\dot{h} \geq - \alpha(h)}$ where ${h \geq 0}$ implies system safety, $\alpha$ is an extended class $\mathcal{K}$ function, and $\dot{h}$ the derivative of $h$ along the solutions of the system. 
The efficacy of this approach has been demonstrated in a variety of applications such as multi-agent systems \cite{lindemann2019control}, robotic manipulators \cite{cortez2020control}, autonomous vessels \cite{thyri2020reactive} and autonomous trucks \cite{Alan__arxiv:22}. 
One of the main challenges in obtaining formal safety guarantees with CBFs in practice is uncertainty: both of the internal model used to synthesize the CBF controller, and the external environment with which the system interacts.  
Since CBFs use models to calculate safe actions, a mismatch between a system and its model can lead degradations in safety \cite{xu2015robustness}.

Robustness in CBF-based methods is typically achieved by introducing a robustifying term in the safety constraint: ${\dot{h} \geq - \alpha(h) + \sigma}$ where the form of $\sigma$ is dictated by the type of uncertainty.  In one of the first works on robust CBFs \cite{jankovic2018robust}, 
the term $\sigma$ was added based upon a bound on the uncertainties with the result being robust safety. 
Later, different observer and identification techniques have been proposed to alleviate the conservativeness of robust controllers by estimating the uncertainty, or considering specific classes of uncertainties \cite{black2022adaptiveKoopman,alan2022dob,lopez2020robust,isaly2021adaptive,cohen2022robust,buch2022robust}. 
Data-driven methods account for uncertainties in a similar fashion where a sufficient condition for the safety is found using properties of uncertainties  \cite{taylor2021data,emam2022data,jin2023data}. 
Learning can also be utilized to estimate the term $\sigma$ in an episodic fashion \cite{taylor2020learning}---this has been deployed successfully on robotic systems \cite{csomay2021episodic}.
In the case of stochastic estimation techniques, such as Gaussian processes, probabilistic safety guarantees are obtained using chance constraints with the standard deviation of the process used as an upper confidence bound \cite{castaneda2021pointwise,akella2022GP}.

%%%%%%%%%%%%%%%%%%%%%%%%%%%%%%%%%%%%%%%%%%%%%%%%%%%%%%%%%%%%%%%%%%%%%%%%
\begin{figure}[t]
	\centering
 \begin{subfloat}
	{\includegraphics[width=.45\textwidth, valign = t]{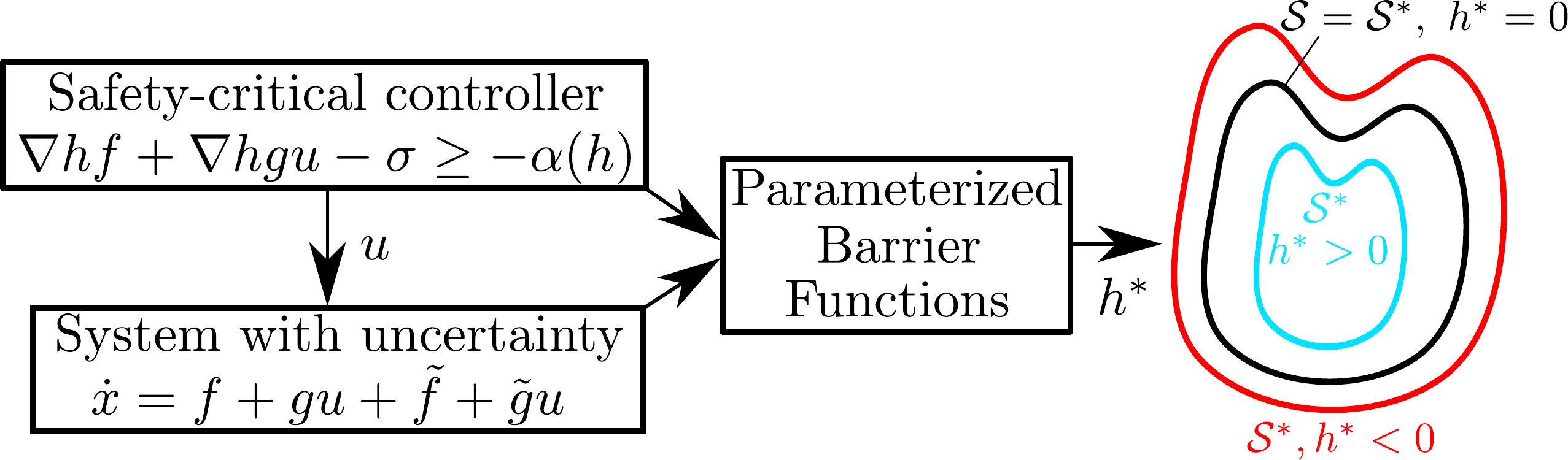}}
\end{subfloat}
	\caption{Illustration of the parameterized barrier (PBF) function framework to obtain safety guarantees for systems with a wide range of uncertainties.}
	\label{fig:intro}
 \vspace{-6 mm}
\end{figure}
%%%%%%%%%%%%%%%%%%%%%%%%%%%%%%%%%%%%%%%%%%%%%%%%%%%%%%%%%%%%%%%%%%%%%%%%

To quantify the effect of uncertainties on safety, it is important to characterize how adding a robustness term, $\sigma$, impacts the ability to satisfy the safety specification, ${h \geq 0}$. Input-to-state safety (ISSf) \cite{kolathaya2018input} provides a means to quantify this relationship.  In particular, in the case of bounded uncertainties, a robustifying $\sigma$ term can be found that enforces safety with respect to the expanded safety specification: ${h \geq -d}$ with some specific ${d > 0}$.  This trade-off between robustness and safety has been elaborated upon with less conservative ISSf conditions that are ``tunable'' \cite{alan2021safe}.
ISSf has proven to be especially useful when implementing CBFs in practice; for example, in the context of safety aware control of quadrupeds \cite{cosner2022safetyaware}, and longitudinal control of full-scale trucks \cite{Alan__arxiv:22}. However, the fact that ISSf ``grows'' the safe set (via ${h \geq -d}$) prevents the analysis of the impacts of uncertainty on the original safe set (${h \geq 0}$). 

The goal of this paper is to generalize the concept of robust CBFs through an extension of the notion of ISSf.  
In particular, to obtain extended safety guarantees, we formulate \emph{parameterized barrier functions (PBFs)}: ${H \triangleq h-h^*}$, where $h$ is the original CBF and $h^*$ is a safety parameter. 
The parameterization of $H$ allows for the relaxation of the strict robust safety condition expressed in the robust CBF formulation. This gives us flexibility to establish safety guarantees for other levels sets of the CBF (via ${h \geq h^*}$ parameterized by $h^*$) in the case that the nominal CBF conditions are not met. 
Importantly, we connect the ISSf framework to PBFs and show that it is possible to obtain improved safety guarantees for ISSf-based controllers.
An inverted pendulum example is used throughout to illustrate the key concepts.

%%%%%%%%%%%%%%%%%%%%%%%%%%%%%%%%%%%%%%%%%%%%%%%%%%%%%%%%%%%%%%%%%%%%%%%%%%%%%%%%%%%%%%%%%%%%%%%%%%%%%%%%%%%%%%%%%%%%%%%%%%%%%%%%%%%%%%%%%%%%%%%%%%%%%%%%%%%%%%%%%%%%%%%%%%%%%%%%%%%%%%%%%%%%%%%%%%%%%%%%%%%%%%%%%%%%%%%%%%%%%%%%%%%%%%%%%%%%%%%%
\section{Background}
\label{sec:CBF}
Consider a nonlinear control system of the form:
\begin{equation}
\label{eq:sys}
    \dot{x}=f(t,x) + g(t,x) u,
\end{equation}
with state ${x\in\R^n}$ and input ${u\in\R^m}$. The functions ${f:\R\times\R^n\to\R^n}$ and ${g:\R\times\R^n\to\R^{n\times m}}$ are locally Lipschitz continuous in $x$ and piece-wise continuous in $t$. A feedback controller ${k:\R\times\R^n\to\R^m}$, ${u=k(t,x)}$, that is locally Lipschitz continuous in $x$ and piece-wise continuous in $t$ implies there exists a time interval ${I(t_0,x_0)\subseteq\R}$ for each initial condition ${x(t_0)=x_0}$ such that the closed loop system has a unique solution $x(t)$ for all $t\in I(t_0,x_0)$ \cite{perko2013differential}. For convenience we take ${t_0=0}$ and assume that the solution exists for all time, that is, ${I(0,x_0)=[0,\infty)}$ for all ${x_0\in\R^n}$. 

Safety is formally defined as the forward invariance of a set in the state space. We define the 0-superlevel set of a continuously differentiable function ${h:\R^n\to\Hs}$, ${\Hs \subseteq \R}$:
\begin{align}
\label{eq:Sc}
    \Sc = \left\{ x\in\R^n ~\left|~ h(x)\geq0 \right.\right\},
\end{align}
such that $\Sc$ is nonempty and has no isolated points,
and we say that the system \eqref{eq:sys} with a controller ${u=k(t,x)}$ is safe with respect to the set $\Sc$ if ${x_0 \in \Sc \implies x(t)\in \Sc}$ for all ${t\geq0}$ and $x_0\in\Sc$.
Control barrier functions~\cite{ames2017control} give us tools to synthesize controllers with safety guarantees. 
\begin{definition}[\!\!\cite{ames2017control}]
    A continuously differentiable function $h$ is a {\em control barrier function  (CBF)} for \eqref{eq:sys} on $\Sc$ if 0 is a regular value\footnote[1]{If ${h(x)=q \implies \nabla h(x)\neq0}$, then $q$ is a regular value of $h$.} and 
    there exists a function $\alpha\in\K_\infty^{\rm e}$\footnote[2]{Function ${\alpha:\R\to\R}$ belongs to \textit{extended class-$\cal{K}_\infty$} (${\alpha\in\cal{K}_{\infty}^{\rm e}}$) if it is continuous, strictly increasing, ${\alpha(0)=0}$, and ${\lim\limits_{r\to\pm\infty}\alpha(r)=\pm\infty}$.}
    such that the following holds for all $t\geq0$ and $x\in\Sc$:
    \begin{align}
    \label{eq:CBF}
        \sup_{u\in\R^m} \left[\nabla h(x) f(t,x) + \nabla h(x) g(t,x)u \right] > -\alpha(h(x)).
    \end{align}
\end{definition}

The existence of a CBF implies that the set of controllers:
\begin{align}
K_{\rm CBF}(t,x) = \{ u\in\R^m  ~|~  &\nabla h(x) f(t,x) \\ &+ \nabla h(x) g(t,x)u \geq - \alpha(h(x))  \} \nonumber
\end{align}
is not empty, and the main result in \cite{ames2017control} states that controllers taking values in $K_{\rm CBF}$ ensure safety:
\begin{theorem}[\!\!\cite{ames2017control}]
    Let $h$ be a CBF for \eqref{eq:sys} on $\Sc$. Then, any controller ${u=k(t,x) \!\in\! K_{\rm CBF}(t,x)}$ renders \eqref{eq:sys} safe w.r.t. $\Sc$.
\end{theorem}

%%%%%%%%%%%%%%%%%%%%%%%%%%%%%%%%%%%%%%%%%%%%%%%%%%%%%%%%%%%%%%%%%%%%%%%%%%%%%%%%%%%%%%%%%%%%%%%%%%%%%%%%%%%%%%%%%%%%%%%%%%%%%%%%%%%%%%%%%%%%%%%%%%%%%%%%%%%%%%%%%%%%%%%%%%%%%%%%
\section{Robust Control Barrier Functions}
\label{sec:RCBF}
Safety guarantees established by CBFs may deteriorate in the presence of an uncertainty in the model. Consider: 
\begin{equation}
\label{eq:sysp}
    \dot{x} = f(t,x) + g(t,x) u + \tilde{f}(t,x) + \tilde{g}(t,x) u,
\end{equation}
where the unknown functions ${\tilde{f}:\R_{\geq0}\times\R^n\to\R^n}$ and ${\tilde{g}:\R_{\geq0}\times\R^n\to\R^{n\times m}}$ are assumed to be locally Lipschitz in $x$ and piece-wise continuous in $t$.
Uncertainties $\tilde{f}$ and $\tilde{g}$ are often called as \emph{additive} and \emph{multiplicative} uncertainties, respectively, emphasizing their relationship with the input $u$ in the dynamics.
Their effect on safety is seen in $\dot{h}$:
\begin{align}
    \dot{h}(t,x,u) =& \overbrace{\nabla h(x) f(t,x) + \nabla h(x) g(t,x) u}^{\dot{h}_{\rm n}(t,x,u)} \nonumber\\
    &+ \underbrace{\nabla h(x) \tilde{f}(t,x)}_{L_{\tilde{f}}h(t,x)} + \underbrace{\nabla h(x) \tilde{g}(t,x)}_{L_{\tilde{g}}h(t,x)}u ,
\label{eq:hdot_p}
\end{align}
where $\dot{h}_{\rm n}$ denotes the \emph{known} portion of $\dot{h}$ while the Lie derivatives $L_{\tilde{f}}h(t,x)$ and $L_{\tilde{f}}h(t,x)$ are unknown.
A controller $u=k(t,x)\in K_{\rm CBF}(t,x)$ yields:
\begin{equation}
    \dot{h}(t,x,k(t,x)) \geq -\alpha(h(x)) + L_{\tilde{f}}h(t,x) + L_{\tilde{g}}h(t,x)k(t,x),
\end{equation}
and no longer satisfies the condition ${\dot{h} \geq - \alpha(h)}$.

%%%%%%%%%%%%%%%%%%%%%%%%%%%%%%%%%%%%%%%%%%%%%%%%%%%%%%%%%%%%%%%%%%%%%%%%
\renewcommand{\arraystretch}{1.0} % Default value: 1
\begin{table*}[t]
\centering
\begin{tabular}{c|cc|c}
~ & ~ & \textbf{Method Summary}   &  $ {\sigma(t,x,u)}$                   \\ \hline \hline

\multirow{18}{*}{\rotatebox[]{90}{RCBF}} & {\cite{jankovic2018robust}}              & Bounded uncertainty: $\|\tilde{f}(t,x)\| \leq {p}$.   & {$ \|\nabla h(x) \|{p}$}   \\ \cline{2-4}

& \multirow{2}{*}{\cite{agrawal2023stateobserver}}        & Bounded uncertainty: $\|\tilde{f}(t,x)\| \leq p$,   & \multirow{2}{*}{$ \kappa(h(x)) \|\nabla h(x) \|p$}   \\ 
& & continuous non-increasing $\kappa$ with $\kappa(0)=1$ & \\ \cline{2-4}

&\multirow{2}{*}{\cite{emam2022data}}    & $\tilde{f}$ is a convex hull of functions $\psi_i(x), i=1,\cdots,q$,  &  $-\min_{i\in\{1,..,q\}} \nabla h(x) \phi_i(x)$   \\ 
                                     &   & $\tilde{g}$ is a convex hull of functions $\rho_i(x), i=1,\cdots,q$,   &  $-\min_{i\in\{1,..,q\}} \nabla h(x) \rho_i(x) u$ \\ \cline{2-4}

&\cite{cohen2022robust} & $[ \tilde{f}, (\tilde{g}~{\rm diag}(u))^\top ]^\top = \psi(x,u)\theta$, and $\exists A,b$ s.t. $A\theta\leq b$   &   $\inf\limits_{A\theta \leq b} \nabla h(x) \psi(x,u) \theta$     \\ \cline{2-4}

&\multirow{2}{*}{\cite{buch2022robust}} & Sector bounded nonlinear perturbation at input,  &  \multirow{2}{*}{$ \left( L_gh(x) - L_{g_{\rm s}}h(x)\right) u + \theta \| u \| \| L_{g_{\rm s}}h(x) \| $}    \\ 
                 &   & i.e., $\exists\alpha,\beta$ defining $g_{\rm s}=\frac{\alpha+\beta}{2}g$ and $\theta=\frac{\beta-\alpha}{\beta+\alpha}$.   &       \\  \cline{2-4}

& \cite{black2022adaptiveKoopman} & $\hat{f}$ estimates $\tilde{f}$ and $b_d$ defines an error band.   & $-\nabla h(x) \hat{f}(x) + b_d(t)$  \\ \cline{2-4}

& \multirow{2}{*}{\cite{alan2022dob}}      & $b(t,x)=L_{\tilde{f}}h(t,x)$ is Lipschitz in $x$ (constant $L_b$), & \multirow{2}{*}{$ -\hat{b}(t,x) + L_b/k_b $} \\ 
                                       &  & $\hat{b}$ estimates $b$ and $k_b$ is estimation gain.   & \\ \cline{2-4}

& \multirow{2}{*}{\cite{isaly2021adaptive}} & $L_{\tilde{f}}h(t,x)=\rho(t,x)\theta$, $\|\theta\|\leq\overline{\theta}$ & \multirow{2}{*}{$ \min \left\{ \|\rho(t,x)\| \overline{\theta}, -\rho(t,x)\hat{\theta}(t) + \|\rho(t,x)\| \tilde{\theta}_{U}(t)               \right\}$} \\ 
                                       &  & $\hat{\theta}$ estimates $\theta$ and $\tilde{\theta}_{U}$ is upper error bound.   & \\ \cline{2-4}

&\multirow{2}{*}{\cite{taylor2021data}} & $\tilde{f}$ and $\tilde{g}$ are Lipschitz in $x$ (with $L_{\tilde{f}}, L_{\tilde{g}}$), $\exists N$   &  $ {\sum\limits_{i=1}^N} \left( \lambda_i^T \tilde{F}_i-\|\lambda_i\| (L_{\tilde{f}}+L_{\tilde{g}}\|u_i\|)\|x-x_i\| \right), $     \\ 
                 &   & data $x_i$, $u_i$ and $\dot{x}_i$ so $\tilde{F}_i=\dot{x}_i-{f}(x_i)+{g}(x_i)u_i$     & $\lambda_i$ are Lagrange multipliers. \\ \cline{2-4}

&\multirow{2}{*}{\cite{jin2023data}} & $\dot{h}_{\rm n}(t,x,u)=0$, $\dot{h}$ is Lipschitz in $x$ and $u$  &  \multirow{2}{*}{$ \min\limits_{i\in[1\cdots N]} \left[-\dot{h}_i + L_x \|x-x_i\| + L_u\|u-u_i\| \right] $} \\
                                &    & (with $L_x,L_u$), $\exists N$ data $x_i$, $u_i$ and $\dot{h}_i$   &     \\ \hline \hline      

\multirow{2}{*}{\rotatebox[]{90}{ISSf}} & \multirow{2}{*}{\cite{alan2021safe}}           & Bounded uncertainty, continuously  & \multirow{2}{*}{$ \dfrac{\|\nabla h(x) \|^2}{\epsilon(h(x))}$}   \\ 
& &   differentiable $\epsilon>0$ with $\der{\epsilon}{r}\geq0, \forall r$  & \\ \hline \hline

\end{tabular}
\caption{A brief summary of robust control barrier function (RCBF) and input-to-state safety (ISSf) based methods for robust safety-critical control, with the corresponding $\sigma$ term used in~\eqref{eq:KRCBF}, to provide safety with robustness against the uncertainties in~\eqref{eq:sysp}.}
\vspace{-5 mm}
\label{tab:sigmas}
\end{table*}
%%%%%%%%%%%%%%%%%%%%%%%%%%%%%%%%%%%%%%%%%%%%%%%%%%%%%%%%%%%%%%%%%%%%%%%%

In the literature this problem is often addressed by adding a compensation term to the safety constraint \eqref{eq:CBF} for robustness against the uncertainty.
To capture this term for a variety of approaches, we generalize the notion of robust CBF, which was first proposed in \cite{jankovic2018robust} using a specific compensation term for a specific type of uncertainty.
\begin{definition}
    A continuously differentiable function $h$ is a {\em robust control barrier function (RCBF)} for \eqref{eq:sysp} on $\Sc$ if 0 is a regular value of $h$ and there exist functions ${\sigma:\R_{\geq0}\times\R^n\times\R^m\to\R}$ and ${\alpha\in\K_\infty^{\rm e}}$ such that the following holds for all ${t\geq0}$ and ${x\in\Sc}$:
\begin{align}
\label{eq:RCBF}
    \sup_{u\in\R^m} \left[ \dot{h}_{\rm n}(t,x,u) - \sigma(t,x,u)\right] > -\alpha(h(x)).
\end{align}
\end{definition}
The compensation term $\sigma$ allows one to cancel the undesired effects of uncertainties on safety.
Similar to CBFs, the existence of a RCBF implies that the set of controllers:
\begin{align}
\resizebox{1\hsize}{!}{
    $K_{\rm RCBF}(t,x) = \{ u\in\R^m  ~|~  \dot{h}_{\rm n}(t,x,u) - \sigma(t,x,u) \geq - \alpha(h(x))  \}$
    }\!
    \label{eq:KRCBF}
\end{align}
is not empty.
Then, the following theorem, generalized from \cite{jankovic2018robust}, gives a sufficient condition to obtain robust safety results:
\begin{theorem}\label{theo:robust}
    Let $h$ be a RCBF for \eqref{eq:sysp} on $\Sc$ with $\sigma$ satisfying:
    \begin{equation}
        \label{eq:idealcomp}
    L_{\tilde{f}}h(t,x) + L_{\tilde{g}}h(t,x) u + \sigma(t,x,u) \geq 0,
    \end{equation}
    for all ${t \geq 0}$, ${x\in\partial\Sc}$ and ${u\in\R^m}$. Then, any controller ${u=k(t,x)\in K_{\rm RCBF}(t,x)}$ renders \eqref{eq:sysp} safe w.r.t. $\Sc$.
\end{theorem}
%Proof is straightforward, thus we skip it here.

% \begin{proof}
%     %To prove that the set $\Sc$ is forward invariant with uncertainties, 
%     We start with $\dot{h}$ given in \eqref{eq:hdot_p}. For any controller ${u=k(t,x)\in K_{\rm RCBF}(t,x)}$ we have:
%     \begin{align}
%     % \resizebox{1\hsize}{!}{
%     % $
%     \dot{h}(t,x,k(t,x)) \geq & -\alpha(h(x)) + \sigma(t,x,k(t,x)) \nonumber \\
%     & + L_{\tilde{g}}h(t,x)k(t,x) + L_{\tilde{f}}h(t,x).
%     % $
%     % }\!.
%     \end{align}
%     Considering $x\in\partial\Sc$, i.e., $h(x)=0$, \eqref{eq:idealcomp} implies that:
%     \begin{align}
%     \dot{h}(t,x,k(t,x)) &\geq 0.
%     \end{align}
%     Since 0 is a regular value of $h$, the forward invariance of $\Sc$ follows from this based on \cite{blanchini2008set}.
% \end{proof}

\begin{remark}
    Robust safety-critical controller design is often formulated as the optimization problem:
    \begin{align}
    \label{eq:QP}
    \smash{k(t,x) =  \,\,\underset{u\in\R^m}{\argmin}}  &  \quad \| u-k_{\rm d}(t,x) \|^2  \\
    \mathrm{s.t.} \quad & \quad \dot{h}_{\rm n}(t,x,u) - \sigma(t,x,u) \geq - \alpha(h(x)), \nonumber
    \end{align}
    where ${k_{\rm d}:\R_{\geq0}\times\R^n\to\R^m}$ is a desired controller. %The form of $\sigma$ specifies the nature of the control problem such as a Quadratic Program \cite{jankovic2018robust} or a Second Order Cone \cite{taylor2021data}. 
\end{remark}

%A plethora of methods has been proposed in the literature to design $\sigma$; see a list in Table~\ref{tab:sigmas}. We take the RCBF framework in \cite{jankovic2018robust} as example to illustrate robust safety. For this particular method, assume there exists ${p}\geq0$ such that:
A plethora of methods has been proposed in the literature to design $\sigma$; see a list in Table~\ref{tab:sigmas} for RCBF-based methods as well as an input-to-state safety-based method that will be described in Section~\ref{sec:ISSf}. 
To illustrate robust safety we now consider a certain class of uncertainty with bounded additive term and no multiplicative term:
\begin{equation}
\label{eq:pbar}
    \|\tilde{f}(t,x)\| \leq {p}, \quad
    \tilde{g}(t,x)=0,
\end{equation}
for all ${t \geq 0}$, ${x \in \R^n}$ with some bound ${p\geq0}$. Inspired by \cite{jankovic2018robust}, we will use a compensation term of the form:
\begin{equation}
\label{eq:sigmaRBCF}
    \sigma(t,x,u) = \| \nabla h(x) \| {{p}},
\end{equation}
which, along with \eqref{eq:pbar}, implies that \eqref{eq:idealcomp} is satisfied, and thus the controller \eqref{eq:QP} with \eqref{eq:sigmaRBCF} keeps the set $\Sc$ safe. Results obtained through Theorem~\ref{theo:robust} are illustrated using an inverted pendulum example with a time varying uncertainty. 

%%%%%%%%%%%%%%%%%%%%%%%%%%%%%%%%%%%%
\begin{example} 
\label{ex:pendulum}

Consider the inverted pendulum in Fig.\ref{fig:IVP}(a) that consists of a massless rod of length $l$ and a concentrated mass $m$.
The pendulum is actuated with a torque $u$, while an unknown external force $F(t)$ is acting horizontally on the mass.
With the angle $\theta$, angular velocity $\dot{\theta}$, and state ${x = \begin{bmatrix}\theta & \dot{\theta} \end{bmatrix}^\top}$, the equation of motion of the pendulum reads:
\begin{equation}
    \label{eq:IVPmodel}
    \dot{x} = \underbrace{\begin{bmatrix} x_2 \\ \frac{g}{l} \sin x_1 \end{bmatrix}}_{f(t,x)} + \underbrace{\begin{bmatrix} 0 \\ \frac{1}{ml^2} \end{bmatrix}}_{g(t,x)} u + \underbrace{\left[ \begin{matrix} 0 \\ \frac{F(t)}{ml}\cos x_1 \end{matrix}   \right]}_{\tilde{f}(t,x)},
\end{equation}
where $g$ is the gravitational acceleration.
All the parameters used in this example are given in Table~\ref{tab:param}. 

The external force $F(t)$ yields an additive uncertainty $\tilde{f}$ ($\tilde{g}(t,x)\equiv0$). We assume that there exists an upper bound $\overline{F}$ such that $|F(t)|\leq\overline{F},~\forall t\geq0$, which yields ${p}=\frac{\overline{F}}{ml}$. A piece-wise continuous force is considered for simulations:
\begin{equation}
\label{eq:ivp_dist}
    F(t) = \overline{F} \big( 1-2s(t-5)+s(t-10)+s(t-15) \big),
\end{equation}
where $s$ is the Heaviside function. 

We seek to design a control torque $u$ such that we keep the pendulum upright within a given safe region of angles, even with the disturbance $F(t)$.
The set $\Sc$ is defined using:
\begin{equation}
\label{eq:IVP_h}
    h(x) = 1 - \frac{1}{2}x^\top A x, \quad
    A =
    \begin{bmatrix}
    2 q_1^2  & q_1 q_2 \\
    q_1 q_2  &2 q_2^2
    \end{bmatrix},
\end{equation}
with parameters ${q_1,q_2>0}$ given in Table~\ref{tab:param}.
The resulting set $\Sc$ is the black ellipse in Fig.~\ref{fig:IVP}(b). Note that ${\nabla h(x)=0}$ only if ${x=0}$, while ${h(0)\!=1}$, thus 0 is a regular value of $h$.

%%%%%%%%%%%%%%%%%%%%%%%%%%%%%%%%%%%%%%%%%%%%%%%%%%%%%%%%%%%%%%%%%%%%%%%%
\begin{table}[b]
    \centering
     \vspace{-4 mm}
    \begin{tabular}{|c|c|c|}
    \hline
     $g=10$ m/s$^2$      & $m=2$ kg                & $l=1$ m         \\ \hline
     $\overline{F}=2$ N  & $q_1=4$ 1/rad           & $q_2=2$ s/rad    \\ \hline
     $\alpha_c=8$ 1/s    & $K_{\rm p}=0.6$ 1/s$^2$ & $K_{\rm d}=0.6$ 1/s   \\ \hline 
    \end{tabular}
    \caption{Parameters used for Example 1.}
    \label{tab:param}
\end{table}
%%%%%%%%%%%%%%%%%%%%%%%%%%%%%%%%%%%%%%%%%%%%%%%%%%%%%%%%%%%%%%%%%%%%%%%%

A desired controller is selected as:
\begin{equation}
    \label{eq:ivp_kn}
    k_{\rm d}(x) = ml^2 \left( -g/l\sin x_1-K_{\rm p}x_1-K_{\rm d}x_2 \right)
\end{equation}
with parameters ${K_{\rm p},K_{\rm d}>0}$. We use \eqref{eq:QP} as robust safety-critical controller with $\sigma$ in \eqref{eq:sigmaRBCF} and ${\alpha(r) = \alpha_c r}$, ${\alpha_c>0}$. Simulation results are depicted in Fig.~\ref{fig:IVP}(b) as a blue curve. The controller successfully keeps the system safe w.r.t. $\Sc$. 

\end{example}
%%%%%%%%%%%%%%%%%%%%%%%%%%%%%%%%%%%%

%%%%%%%%%%%%%%%%%%%%%%%%%%%%%%%%%%%%%%%%%%%%%%%%%%%%%%%%%%%%%%%%%%%%%%%%
\begin{figure}[t]
	\centering
    \includegraphics[width=0.45\textwidth, valign = c]{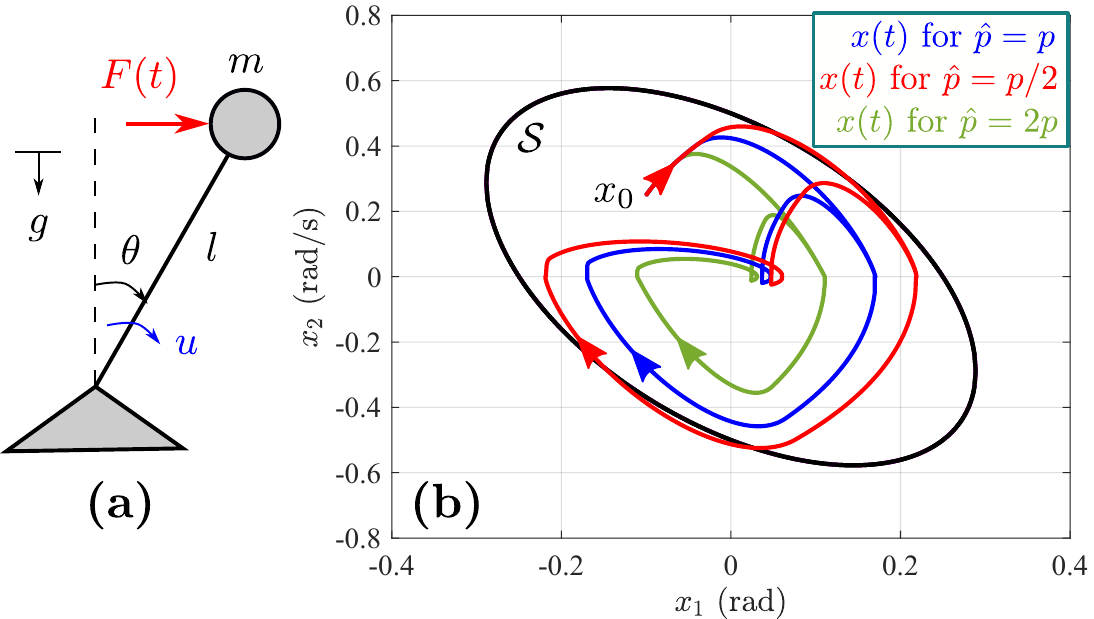}
	\caption{(a) Inverted pendulum model. (b) The safe set $\Sc$ (black ellipse) and simulated trajectories (colored curves) with the controller \eqref{eq:QP} for different values of the estimated uncertainty bounds; cf. \eqref{eq:sigmaRBCF} and \eqref{eq:sigmaRBCFapprox}.}
	\label{fig:IVP}
 \vspace{-7 mm}
\end{figure}
%%%%%%%%%%%%%%%%%%%%%%%%%%%%%%%%%%%%%%%%%%%%%%%%%%%%%%%%%%%%%%%%%%%%%%%%

To achieve robust safety, the compensation term $\sigma$ is typically designed based on certain properties of $\tilde{f}$ and $\tilde{g}$ such as the upper bound ${p}$ in \eqref{eq:pbar}; see Table~\ref{tab:sigmas}. In practice, these properties may be hard to estimate, thus the compensation \eqref{eq:idealcomp} required for robust safety may not be realized. 
For example, if ${p}$ in \eqref{eq:pbar} is not known precisely, one may rely on an estimation $\hat{p}$ of $p$ instead, with the compensation term:
    \begin{equation}
    \label{eq:sigmaRBCFapprox}
        \sigma(t,x,u) = \| \nabla h(x) \| \hat{p}.
    \end{equation}
Then, under-approximating the size of the uncertainty may yield safety degradation, while over-approximation may 
induce conservative behavior that is not captured by Theorem~\ref{theo:robust}.
We illustrate these for the inverted pendulum problem.

%%%%%%%%%%%%%%%%%%%%%%%%%%%%%%%%%%%%
\begin{example} 
\label{ex:robust}
    Consider the system in Example~\ref{ex:pendulum} and \eqref{eq:sigmaRBCFapprox}.
    If the uncertainty is under-approximated (${\hat{p}<p}$), \eqref{eq:idealcomp} is not satisfied and Theorem~\ref{theo:robust} cannot establish safety guarantees. Indeed, simulations capture safety degradation where $x(t)$ leaves $\Sc$; see the red curve in Fig.~\ref{fig:IVP}(b) for ${\hat{p}={p}/2}$.
    If the uncertainty is over-approximated (${\hat{p}>p}$), \eqref{eq:idealcomp} holds and Theorem~2 implies that the set $\Sc$ is safe.
    The corresponding simulation results, depicted in Fig.~\ref{fig:IVP}(b) as a green curve for ${\hat{p}=2{p}}$, comply with this.
    However, we observe conservative behavior where $x(t)$ evolves inside a smaller subset of $\Sc$.
\end{example}
%%%%%%%%%%%%%%%%%%%%%%%%%%%%%%%%%%%%

%%%%%%%%%%%%%%%%%%%%%%%%%%%%%%%%%%%%%%%%%%%%%%%%%%%%%%%%%%%%%%%%%%%%%%%%%%%%%%%%%%%%%%%%%%%%%%%%%%%%%%%%%%%%%%%%%%%%%%%%%%%%%%%%%%%%%%%%%%%%%%%%%%%%%%%%%%%%%%%%%%%%%%%%%%%%%%%%%%%%
\section{Parameterized Barrier Functions}
\label{sec:FBF}
To quantify safety degradation and conservativeness emerging from non-ideal compensation of uncertainties, we extend the RCBF-based safety guarantees by introducing the concept of {\em parameterized barrier function}.

Our key idea is to establish safety guarantees for other superlevel sets of $h$ than $\Sc$.
Thus, we introduce the set:
\begin{align}
\label{eq:Sstar}
    \Sc^* &\triangleq \left\{ x\in\R^n ~\left|~ H(x,h^*)\geq0 \right.\right\}, \\
    \partial\Sc^* &\triangleq \left\{ x\in\R^n ~\left|~ H(x,h^*)=0 \right.\right\}, 
\end{align}
where the function $H:\R^n\times\Hs\to\R$ is given as: 
\begin{equation}
\label{eq:FBF}
    H(x,h^*)\triangleq h(x)-h^*;
\end{equation}
with $h$ defining the set $\Sc$ in \eqref{eq:Sc} and a parameter ${h^* \in \Hs}$ to be determined. 
We have ${\Sc^*\supset\Sc}$ if ${h^*<0}$, ${\Sc^*=\Sc}$ if ${h^*=0}$, and ${\Sc^*\subset\Sc}$ if ${h^*>0}$; see Fig.~\ref{fig:intro}. We assume that the set $\Sc^*$ is nonempty and has no isolated points for any ${h^*\in\Hs}$.

\begin{definition}
    Function $H$ is a {\em parameterized barrier function (PBF)} for \eqref{eq:sysp} on $\Sc^*$ if $h$ is a RCBF for \eqref{eq:sysp} on $\Sc^*$ and $h^*$ is a regular value of $h$.
\end{definition}

The following theorem presents the conditions for safety of \eqref{eq:sysp} w.r.t. $\Sc^*$, and ultimately allows us to characterize safety degradation and conservativeness.
\begin{theorem} \label{theo:flexible}
    Let $H$ be a PBF for \eqref{eq:sysp} on $\Sc^*$ with ${h^*\in\Hs}$ and $\sigma$ satisfying:
    \begin{equation}
        \label{eq:FBF_cond}
    L_{\tilde{f}}h(t,x) + L_{\tilde{g}}h(t,x) u + \sigma(t,x,u) \geq \alpha(h^*),
    \end{equation}
    for all ${t \geq 0}$, ${x\in\partial\Sc^*}$ and ${u\in\R^m}$. Then, any controller ${u=k(t,x)\in K_{\rm RCBF}(t,x)}$ renders~\eqref{eq:sysp} safe w.r.t. $\Sc^*$.
\end{theorem}
\begin{proof}
    $H$ is continuously differentiable since $h$ is a RCBF and $h^*$ is a constant, and we have:
    \begin{align}
        \dot{H}(t,x,u)
        \!=\!\dot{h}(t,x,u)
        \!=\!\dot{h}_{\rm n}(t,x,u)\!+\!L_{\tilde{f}}h(t,x)\!+\!L_{\tilde{g}}h(t,x)u.
    \end{align}
    For any controller ${u=k(t,x)\in K_{\rm RCBF}(t,x)}$ this yields:
    \begin{align}
    \dot{H}(t,x,k(t,x)) \geq -\alpha(h(x)) + \sigma(t,x,k(t,x))
    \nonumber \\
    + L_{\tilde{f}}h(t,x) + L_{\tilde{g}}h(t,x)k(t,x).
    \end{align}
    Considering $x\in\partial\Sc^*$, i.e., $h(x)=h^*$, \eqref{eq:FBF_cond} implies that
    \begin{align}
    \dot{H}(t,x,k(t,x)) &\geq 0.
    \end{align}
    Since $h^*$ is a regular value of $h$ we have that 0 is a regular value of $H$. Thus, the rest of the proof follows from \cite{blanchini2008set}. 
\end{proof}

It is noted that Theorem~\ref{theo:robust} is a special case of Theorem~3 with ${h^*=0}$.
Yet, Theorem~\ref{theo:flexible} has two main contributions over Theorem~\ref{theo:robust}, thanks to its parameterization by $h^*$. If \eqref{eq:FBF_cond} holds with ${h^*<0}$, condition \eqref{eq:idealcomp} may not hold and Theorem~\ref{theo:robust} cannot establish safety.
Still, Theorem~\ref{theo:flexible} provides  safety guarantees w.r.t. the set ${\Sc^* \supset \Sc}$, hence it quantifies safety degradation.
If \eqref{eq:FBF_cond} holds with ${h^*>0}$, condition \eqref{eq:idealcomp} also holds, and Theorem~\ref{theo:robust} establishes safety w.r.t. $\Sc$.
However, Theorem~\ref{theo:flexible} also states safety w.r.t. the set ${\Sc^* \subset \Sc}$, hence it quantifies conservativeness.
This is summarized for the case of the compensation term in~\eqref{eq:sigmaRBCFapprox} as follows.

\begin{corollary} \label{cor:flexible_rcbf}
Consider~\eqref{eq:sysp} with~\eqref{eq:pbar} and~\eqref{eq:sigmaRBCFapprox}.
Assume that there exist ${\underline{\delta}, \overline{\delta}: \Hs \to \R_{\geq 0}}$
such that for any ${h^* \in \Hs}$:
\begin{align}
\label{eq:delta}
    \underline{\delta}(h^*) \leq
    \|\nabla h(x)\|
    \leq \overline{\delta}(h^*),
\end{align}
${\forall x\in\partial\Sc^*}$.
If $H$ is a PBF for \eqref{eq:sysp} on $\Sc^*$ with $h^*$ defined by:
\begin{equation}
\label{eq:hstar_rcbf}
    \alpha(h^*) =
    \begin{cases}
        \overline{\delta}(h^*)(\hat{p}-p) & {\rm if}\ \hat{p}<p, \\
        \underline{\delta}(h^*)(\hat{p}-p) & {\rm if}\ \hat{p}>p,
    \end{cases}
\end{equation}
then ${u=k(t,x)\in K_{\rm RCBF}(t,x)}$ renders~\eqref{eq:sysp} safe w.r.t. $\Sc^*$.
\end{corollary}

\begin{proof}
The choice~\eqref{eq:sigmaRBCFapprox} of the robustifying term $\sigma$ implies:
\begin{equation}
\label{eq:ex2_temp}
    L_{\tilde{f}}h(t,x) + L_{\tilde{g}}h(t,x)u + \sigma(t,x,u) \geq \|\nabla h(x)\| (\hat{p}-p).
\end{equation}
This leads to~\eqref{eq:FBF_cond} by using~\eqref{eq:delta} and~\eqref{eq:hstar_rcbf}, and the rest of the proof follows from Theorem~\ref{theo:flexible}.
\end{proof}

\begin{remark}
The value of $h^*$ given by~\eqref{eq:hstar_rcbf} quantifies safety degradation and conservativeness.
If the uncertainty is under-approximated (${\hat{p}<p}$), \eqref{eq:hstar_rcbf} yields ${h^*<0}$ and ${\Sc^* \supset \Sc}$, while over-approximation (${\hat{p}>p}$) leads to ${h^*>0}$ and ${\Sc^* \subset \Sc}$.
\end{remark}

%%%%%%%%%%%%%%%%%%%%%%%%%%%%%%%%%%%%
\begin{example}
\label{ex:IVP_FBF}
Consider the setup of Example~\ref{ex:robust}.
Based on \eqref{eq:IVP_h}, we get ${\nabla h(x) = -Ax}$,  and
it can be shown that \eqref{eq:delta} holds for any ${h^* \in \Hs = (-\infty,1]}$ and for all ${x\in\partial\Sc^*}$
with
${\underline{\delta}(h^*) = \sqrt{2 \underline{\lambda} (1-h^*)}}$
and
${\overline{\delta}(h^*) = \sqrt{2 \overline{\lambda} (1-h^*)}}$,
where ${0 < \underline{\lambda}\leq \overline{\lambda}}$ are the eigenvalues of $A$.
Since any ${h^*<1}$ is a regular value of $h$,
Corollary~\ref{cor:flexible_rcbf} establishes safety w.r.t. the set $\Sc^*$ with $h^*$ given by~\eqref{eq:hstar_rcbf}.

The value of $h^*$ is depicted in Fig.~\ref{fig:FBF1} with dashed line along with $h(x(t))$ corresponding to the simulated trajectories in Fig.~\ref{fig:IVP}(b). Observe that for the case of under-approximation (${\hat{p}<p}$, red) the PBF framework successfully quantifies safety degradation by a lower bound ${h^*<0}$ for $h(x(t))$, which complies with the simulation results. For the case of over-approximation (${\hat{p}>p}$, green) the bound ${h^*>0}$ captures the safe but conservative system behavior.

%%%%%%%%%%%%%%%%%%%%%%%%%%%%%%%%%%%%%%%%%%%%%%%%%%%%%%%%%%%%%%%%%%%%%%%%
\begin{figure}
	\centering
    \includegraphics[width=.39\textwidth, valign = t]{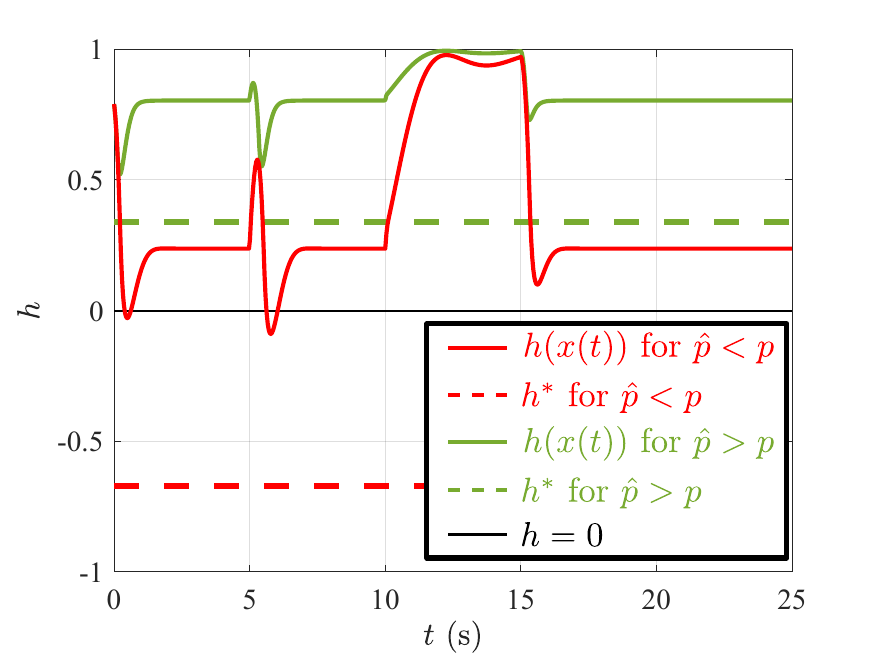}
	\caption{The time evolution of the RCBF, $h(x(t))$, for the simulations of Example~\ref{ex:robust}, and the safety guarantees established in Example~\ref{ex:IVP_FBF} using the PBF framework in the form of a lower bound $h^*$ for $h(x(t))$.}
	\label{fig:FBF1}
\vspace{-6 mm}
\end{figure}
%%%%%%%%%%%%%%%%%%%%%%%%%%%%%%%%%%%%%%%%%%%%%%%%%%%%%%%%%%%%%%%%%%%%%%%%

\end{example}
%%%%%%%%%%%%%%%%%%%%%%%%%%%%%%%%%%%%

%%%%%%%%%%%%%%%%%%%%%%%%%%%%%%%%%%%%%%%%%%%%%%%%%%%%%%%%%%%%%%%%%%%%%%%%%%%%%%%%%%%%%%%%%%%%%%%%%%%%%%%%%%%%%%%%%%%%%%%%%%%%%%%%%%%%%%%%%%%%%%%%%%%%%%%%%%%%%%%%%%%%%%%%%%%%%%%%%%%%%%%%%%%%%%%%%%%%%%%%
\section{Input-to-State Safety \\ via Parameterized Barrier Functions}
\label{sec:ISSf}
A well-known existing concept proposed to characterize safety degradation is input-to-state safety (ISSf) \cite{kolathaya2018input}\footnote[3]{Although ISSf was originally proposed for matched input disturbances, in this study we extend it for additive type of uncertainties $\tilde{f}$.}. 
In this section we show that ISSf is a special case of
the PBF framework (restricted to ${h^*<0}$). 
Then, we propose a method that endows ISSf-CBF-based controllers with more accurate safety guarantees (including ${h^*=0}$ and ${h^*>0}$). 

In essence, ISSf gives ways to quantify safety degradation in the presence of a bounded disturbance such as \eqref{eq:pbar}. This inspired PBFs, as ISSf considers safety degradation in the context of safety guarantees for another superlevel set of $h$:
\begin{align}
\label{eq:Sc_ISSf}
    \Sc_{\rm ISSf} \!=\! \left\{ x\in\R^n  \left|  h(x) - \alpha^{-1}\left( -\frac{\epsilon(h(x)) {p}^2}{4} \right) \!\geq\! 0 \right.\right\},
\end{align}
with a continuously differentiable function ${\epsilon:\Hs\to\R_{>0}}$ that satisfies
${\der{\epsilon}{r}(r) \geq 0}$, ${\forall r\in\Hs}$ and ${\alpha^{-1}\in\K_\infty^{\rm e}}$ \cite{alan2021safe}. 
ISSf-CBFs provide controllers with safety guarantees w.r.t. $\Sc_{\rm ISSf}$:
\begin{definition}
    A continuously differentiable function $h$ is an {\em input-to-state safe control barrier function (ISSf-CBF)} for \eqref{eq:sysp} if there exist a function $\alpha\in\K_\infty^{\rm e}$ such that the following holds for all $t\geq0$ and $x\in\R^n$:
\begin{align}
\label{eq:ISSfCBF}
    \sup_{u\in\R^m} \left[ \dot{h}_{\rm n}(t,x,u)\right] > -\alpha(h(x)) + \frac{\|\nabla h(x) \|^2}{\epsilon(h(x))}.
\end{align}
\end{definition}

\begin{remark}
\label{rem:tunability}
While the original ISSf formulation in \cite{kolathaya2018input} considers ${\der{\epsilon}{r}(r) = 0}$, our work in \cite{Alan__arxiv:22} shows that controller performance can be improved by choosing ${\der{\epsilon}{r}(r) > 0}$ through experiments with an automated truck. 
\end{remark}

Theorem~3 in \cite{alan2021safe} establishes safety for \eqref{eq:sysp} w.r.t. $\Sc_{\rm ISSf}$, if the controller takes values in the non-empty set:
\begin{align}
\resizebox{1\hsize}{!}{
    $K_{\rm ISSf}(t,x) = \big\{ u\in\R^m  ~|~  \dot{h}_{\rm n}(t,x,u) \geq - \alpha(h(x))  + \frac{\|\nabla h(x) \|^2}{\epsilon(h(x))} \big\}$.
    }\!
    \label{eq:KISSf}
\end{align}
In the next theorem, we link ISSf-CBFs to the PBF framework and establish the same result via PBFs.
\begin{theorem}
    \label{theo:ISSf}
    If $h$ is an ISSf-CBF for \eqref{eq:sysp} with \eqref{eq:pbar}, then $H$ is a PBF for this system on $\Sc^*$ with:
    \begin{equation}
    \label{eq:sigmaISSf}
        \sigma(t,x,u)=\frac{\|\nabla h(x)\|^2}{\epsilon(h(x))},
    \end{equation}
    and $h^*$ defined by:
    \begin{equation}
    \label{eq:ISSf_temp2}
        h^* = \alpha^{-1} \left( -{\epsilon(h^*) {p}^2}/{4}  \right) < 0.
    \end{equation} 
    Furthermore, any controller ${u=k(t,x)\in K_{\rm ISSf}(t,x)}$ renders~\eqref{eq:sysp} safe w.r.t. ${\Sc^*=\Sc_{\rm ISSf}\supset \Sc}$.
\end{theorem}
\begin{proof}
    First, we observe that \eqref{eq:ISSf_temp2} has a unique solution $h^*$ based on the monotonicity properties of $\alpha^{-1}$ and $\epsilon$.
    Furthermore, ${\Sc^*=\Sc_{\rm ISSf}}$ based on \eqref{eq:Sc_ISSf} and \eqref{eq:ISSf_temp2}, while the property ${\epsilon(r)>0}$ for all ${r\in\Hs}$ yields ${h^*<0}$ and ${\Sc^* \supset \Sc}$.
    Moreover, $h^*$ is a regular value of $h$ thanks to the strict inequality in \eqref{eq:ISSfCBF}; please refer to the proof of Theorem~3 in \cite{alan2021safe} for details.
    Hence, comparing \eqref{eq:ISSfCBF} with \eqref{eq:RCBF} and \eqref{eq:sigmaISSf} establishes that $H$ is a PBF.
    Finally, by noticing that ${\frac{\|\nabla h(x)\|^2}{\epsilon(h(x))} - \|\nabla h(x) \| {p} \geq -\frac{\epsilon(h(x)) {p}^2}{4}}$, \eqref{eq:pbar} and \eqref{eq:sigmaISSf} yield:
    \begin{align}
    \label{eq:ISSf_temp}
        L_{\tilde{g}}h(t,x)u + L_{\tilde{f}}h(t,x) + \sigma(t,x,u) \!\geq\! -{\epsilon(h(x)) {p}^2}/{4}.
    \end{align}
    This inequality and \eqref{eq:ISSf_temp2} imply that condition \eqref{eq:FBF_cond} in Theorem~\ref{theo:flexible} holds, therefore \eqref{eq:sysp} is safe w.r.t. $\Sc^*$.
\end{proof}

Next, we derive more accurate safety guarantees for ISSf-CBF-based controllers via the PBF framework.
\begin{corollary}
Consider~\eqref{eq:sysp} with~\eqref{eq:pbar} and~\eqref{eq:sigmaISSf}.
Assume that there exists ${\underline{\delta}: \Hs \to \R_{\geq 0}}$
such that for any ${h^* \in \Hs}$:
\begin{align}
\label{eq:deltalow}
    \underline{\delta}(h^*) \leq
    \|\nabla h(x)\|,
\end{align}
${\forall x\in\partial\Sc^*}$.
If $H$ is a PBF for \eqref{eq:sysp} on $\Sc^*$ with $h^*$ defined by:
\begin{equation}
\label{eq:FBF_temp3}
    h^* = \alpha^{-1} \left( {\underline{\delta}(h^*)^2}/{\epsilon(h^*)} - \underline{\delta}(h^*) {p}  \right)
\end{equation}
and ${\epsilon(h^*) \leq 2\underline{\delta}(h^*)/p}$ holds,
then ${u\!=\!k(t,x)\!\in\! K_{\rm RCBF}(t,x)}$ renders~\eqref{eq:sysp} safe w.r.t. ${\Sc^* \subseteq \Sc_{\rm ISSf}}$.
\end{corollary}
\begin{proof}
    Based on \eqref{eq:pbar}, \eqref{eq:sigmaISSf} and ${\epsilon(h^*) \leq 2\underline{\delta}(h^*)/p}$, it can be shown that the following holds for all $x \in \partial \Sc^*$:
    \begin{equation}
    \label{eq:FBF_temp}
        L_{\tilde{g}}h(t,x)u + L_{\tilde{f}}h(t,x) + \sigma(t,x,u) \geq \frac{\underline{\delta}(h^*)^2}{\epsilon(h^*)} - \underline{\delta}(h^*) {p}.
    \end{equation}
    Using~\eqref{eq:FBF_temp3} this leads to~\eqref{eq:FBF_cond},  and Theorem~\ref{theo:flexible} establishes safety w.r.t. $\Sc^*$.
    Moreover, ${h^* \geq \alpha^{-1} \big( -\frac{\epsilon(h^*) {p}^2}{4}  \big)}$ holds, thus ${\Sc^* \subseteq \Sc_{\rm ISSf}}$. 
\end{proof}

\begin{remark}
    Since ${\Sc^* \subseteq \Sc_{\rm ISSf}}$, the PBF framework provides a tighter safety guarantee than ISSf theory. 
    Indeed, all cases of ${h^*<0}$, ${h^*=0}$ and ${h^*>0}$ can occur in \eqref{eq:FBF_temp3}, corresponding to safety degradation, safety and conservativeness. 
\end{remark}

%%%%%%%%%%%%%%%%%%%%%%%%%%%%%%%%%%%%
\begin{example}

%%%%%%%%%%%%%%%%%%%%%%%%%%%%%%%%%%%%%%%%%%%%%%%%%%%%%%%%%%%%%%%%%%%%%%%%
\begin{figure}[t]
	\centering
    \includegraphics[width=.4\textwidth, valign = t]{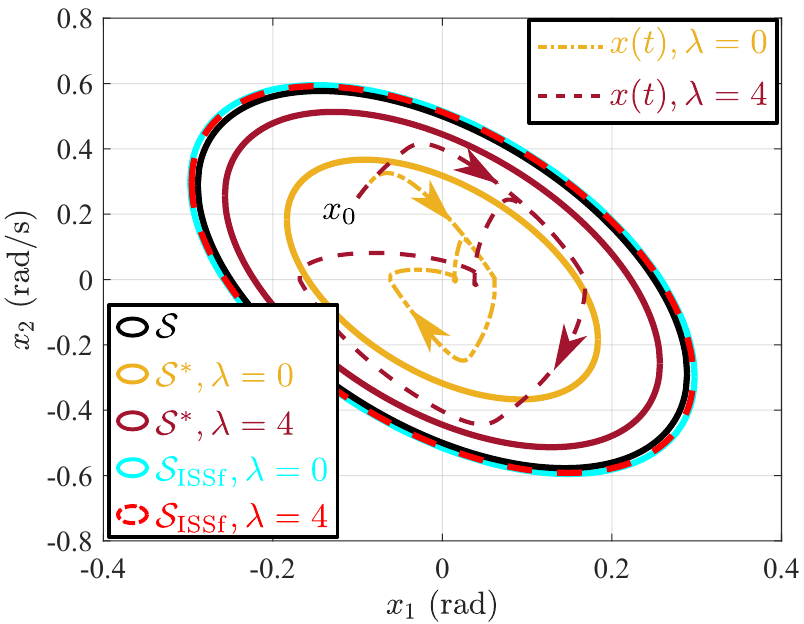}
	\caption{Safety guarantees established in Example~\ref{ex:ISSf} for ${\epsilon_0=1}$ with ${\lambda=0}$ and ${\lambda=4}$ using the ISSf approach ($\Sc_{\rm ISSf}$) and the PBF framework ($\Sc^*$).
    While the ISSf case (cyan and red ellipses) does not capture the conservative behavior,
    the PBF framework yields more accurate safety guarantees (orange and brown ellipses) for simulated trajectories (orange and brown curves). }
	\label{fig:ISSf}
 \vspace{-6 mm}
\end{figure}
%%%%%%%%%%%%%%%%%%%%%%%%%%%%%%%%%%%%%%%%%%%%%%%%%%%%%%%%%%%%%%%%%%%%%%%%

\label{ex:ISSf}
Consider the inverted pendulum problem in Example~\ref{ex:pendulum}. 
We utilize the controller \eqref{eq:QP} with $\sigma$ in \eqref{eq:sigmaISSf},
${\epsilon(r)=\epsilon_0{\rm e}^{\lambda r}}$,
${\epsilon_0>0}$ and ${\lambda\geq0}$. 
Two simulation results are given in Fig.~\ref{fig:ISSf}, with ${\epsilon_0=1}$, ${\lambda=0}$ (orange dashed-dotted curve), and ${\epsilon_0=1}$, ${\lambda=4}$ (brown dashed curve).
Both simulated trajectories stay within $\Sc$. 
Indeed, while the former parameter pair yields a more conservative result, introducing $\lambda$ alleviates the conservativeness as discussed in Remark~\ref{rem:tunability}. 

Boundaries of the corresponding $\Sc_{\rm ISSf}$ sets, calculated by solving \eqref{eq:ISSf_temp2}, are also plotted by cyan solid and red dashed ellipses. As expected, these sets obtained from the ISSf theory fail to evaluate the conservativeness.
The boundaries of the sets $\Sc^*$, after solving \eqref{eq:FBF_temp3}, are plotted by orange and brown solid lines in Fig.~\ref{fig:ISSf}.
Indeed, they are more accurate bounds on the trajectories of the system. 
This shows that the PBF framework provides flexibility to quantify conservativeness.

\end{example}
%%%%%%%%%%%%%%%%%%%%%%%%%%%%%%%%%%%%

%%%%%%%%%%%%%%%%%%%%%%%%%%%%%%%%%%%%%%%%%%%%%%%%%%%%%%%%%%%%%%%%%%%%%%%%%%%%%%%%%%%%%%%%%%%%%%%%%%%%%%%%%%%%%%%%%%%%%%%%%%%%%%%%%%%%%%%%%%%%%%%%%%%%%%%%%%%%%%%%%%%%%%%%%%%%%%%%%%%%%%%%%%%%%%%%%%%%%%%%
\vspace{-2 mm}
\section{Conclusion}
\label{sec:conclusion}
This work focused on establishing safety guarantees for control systems with uncertainties.
We proposed parameterized barrier functions (PBFs) that generalize existing robust control barrier function (RCBF) formulations addressing robust safety.
We highlighted that the PBF framework offers flexibility  to evaluate not only safety, but safety degradation and conservativeness of RCBF-based controllers. 
Moreover, we showed that input-to-state safety (ISSf) can be viewed as a special case of the PBF framework, and we derived improved safety guarantees for ISSf-CBF-based controllers.
Future directions include merging the PBF framework with data-driven schemes to obtain online safety guarantees.

%\vspace{-1 mm}
\bibliographystyle{IEEEtran}
\bibliography{Bib/Alan_bib}

\end{document}